\documentclass[11p,reqno]{amsart}

\usepackage[T1]{fontenc}
\usepackage{graphicx}
\usepackage{amssymb,amsthm,amsmath,mathrsfs,bm,braket,marginnote}
\usepackage{enumerate}
\usepackage{appendix}
\usepackage{extarrows}
\usepackage[colorlinks=true, pdfstartview=FitV, linkcolor=blue, citecolor=blue, urlcolor=blue]{hyperref}
\usepackage{multirow}

\usepackage{pgf}
\usepackage{pgfplots}
\usepackage{tikz}
\usetikzlibrary{arrows,calc}
\usepackage{verbatim}
\usetikzlibrary{decorations.pathreplacing,decorations.pathmorphing}
\usepackage[numbers,sort&compress]{natbib}
\usepackage{dsfont}

\numberwithin{equation}{section}
\newtheorem{theorem}{Theorem}[section]

\newtheorem{remark}[theorem]{Remark}
\newtheorem{corollary}[theorem]{Corollary}

\newtheorem{proposition}[theorem]{Proposition}

 \reversemarginpar

 \newcommand{\Pf}{\text{Pf}}
 \newcommand{\al}{\alpha}
 \newcommand{\alm}{\alpha^{(m)}}
 \newcommand{\be}{\beta}
 \newcommand{\bem}{\beta^{(m)}}
 \newcommand{\adm}{\langle \mathcal{A}\rangle}
 
\newcommand{\sumsigma}{\sum_{\sigma}\text{sgn}\sigma}
\newcommand{\ho}{h^{(1)}}
\newcommand{\hf}{h^{(4)}}
\newcommand{\xim}{\xi^{(m)}}
\newcommand{\etam}{\eta^{(m)}}
\newcommand{\V}{V_1,V_2}
\newcommand{\I}{\tilde{I}}

 \reversemarginpar

\begin{document}

\title[Classical skew orthogonal polynomials in a two-component log-gas]{Classical skew orthogonal polynomials in a two-component log-gas with charges $+1$ and $+2$}

\author{Peter J. Forrester}
\address{School of Mathematical and Statistics, ARC Centre of Excellence for Mathematical and Statistical Frontiers, The University of Melbourne, Victoria 3010, Australia}
\email{pjforr@unimelb.edu.au}

\author{Shi-Hao Li}
\address{ School of Mathematical and Statistics, ARC Centre of Excellence for Mathematical and Statistical Frontiers, The University of Melbourne, Victoria 3010, Australia}
\email{shihao.li@unimelb.edu.au}

\date{}

\dedicatory{}


\begin{abstract}
 There is a two-component log-gas system with Boltzmann
factor which provides an interpolation between the eigenvalue PDF
for $\beta = 1$ and $\beta = 4$ invariant random matrix ensembles.
Its solvability relies on the construction of
particular skew orthogonal polynomials, with the skew inner product
a linear combination of the $\beta = 1$ and $\beta = 4$ inner products,
each involving weight functions. For suitably related classical weight functions,
we seek to express the skew orthogonal polynomials as linear combinations of
the underlying orthogonal polynomials. It is found that in each case
(Gaussian, Laguerre, Jacobi and generalised Cauchy) the coefficients can be
expressed in terms of hypergeometric polynomials with argument relating to the
fugacity. In the Jacobi
case, for example,  the coefficients are Wilson polynomials.
\end{abstract}
\maketitle

\section{Introduction}
\subsection{The circular ensembles and two-component log-gases}
It has been emphasized by Dyson \cite{dyson62} that for general $\beta>0$, the probability density function (PDF) on the unit circle proportional to
\begin{align}\label{1.1}
\prod_{1\leq j<k\leq N}|e^{i\theta_k}-e^{i\theta_j}|^\beta
\end{align}
can be interpreted as the Boltzmann factor for a classical statistical mechanical system with $N$ particles on a unit circle interacting pairwise via the potential $-\log|e^{i\theta_j}-e^{i\theta_k}|$.
This system is usually referred to as the one-component log-gas. Later on, Sutherland \cite{sutherland71} considered the Hamiltonian for the quantum many body system
\begin{align}\label{1.1a}
H=-\sum_{j=1}^N\frac{\partial^2}{\partial \theta_j^2}+\frac{\beta}{4}\Big (\frac{\be}{2}-1 \Big )\sum_{1\leq j<k\leq N}\frac{1}{\sin^2\pi(\theta_k-\theta_j)/2}
\end{align}
and showed that up to normalisation, (\ref{1.1}) is the absolute value squared of the corresponding ground state wave function.

Dyson's concern in  \cite{dyson62} was with the cases $\beta = 1, 2$ and 4 in the context of random matrix theory. In the case $\beta = 2$ it had been shown earlier by Weyl that (\ref{1.1}) corresponds
to the eigenvalue PDF of matrices from the unitary group $U(N)$ chosen with Haar measure. Dyson showed that the eigenvalue PDF of matrices $U U^T$, with
$U \in U(N)$ is given by the case $\beta = 1$. He showed too that for $U \in U(2N)$ the eigenvalue PDF of matrices $U U^D$, where $D$ denotes the dual specified by
$U^D = Z_{2N} U^T Z_{2N}^{-1}$ with $Z_{2N} = I_N \otimes i \sigma_y$ ($\sigma_y$ the appropriate Pauli matrix), is given by the case $\beta = 4$. Here the spectrum
is doubly degenerate and the PDF refers only to the independent eigenvalues. See the recent review \cite{DF17} for more details and motivations.

Starting with (\ref{1.1}) in the cases $\beta = 1, 2$ and 4, a different generalisation to regarding $\beta$ as continuous
is to consider a two-component interpolation between $\beta=1$ and $4$. This is motivated by the fact that the logarithmic potential is the solution of the
Poisson equation in two-dimensions, allowing the particles in the classical statistical mechanical interpretation of (\ref{1.1}) to
be thought of as charges. With the inverse temperature fixed to unity, the particles in (\ref{1.1}) for $\beta = 1$ have charge $+1$, while those in
(\ref{1.1}) for $\beta = 4$ have charge $+2$ --- the pair potential between particles of charges $q_1$ and $q_2$ confined to the unit
circle embedded in two-dimensions is $- q_1 q_2 \log | e^{i \theta_1} - e^{i \theta_2} |$. A natural two-component generalisation is thus
 the PDF proportional to \cite{forrester84}
\begin{align}\label{1.2}
\prod_{1\leq j<k\leq N_1}|e^{i\theta_k}-e^{i\theta_j}|\prod_{1\leq\al<\be\leq N_2}|e^{i\phi_\al}-e^{i\phi_\be}|^4\prod_{j=1}^{N_1}\prod_{\al=1}^{N_2}|e^{i\theta_j}-e^{i\phi_\al}|^2,
\end{align}
which can be viewed as the Boltzmann factor for a two-component log-gas on the unit circle, with $N_1$ particles of charge $+1$ and $N_2$ particles of charge $+2$.  
Since the charges are of the same sign, this two-component log-gas may also be referred to as a two-component plasma.
It is immediate that when $N_2=0$ (resp. $N_1=0$), this model 
corresponds to (\ref{1.1}) with $\beta=1$ (resp. $\beta=4$). It is furthermore the case that (\ref{1.2}) corresponds to the ground state wave function of a two-component generalisation of the
quantum many body Hamiltonian (\ref{1.1a}) \cite{Fo84a}.

Explicit evaluations of the partition function and the two-point charge--charge correlation functions corresponding to (\ref{1.2}) were given in \cite{forrester84}. These explicit formulas allowed
the free energy in the thermodynamic limit to be computed, as well as the evaluation of the bulk scaling limit of the two-point charge--charge correlation functions; see also \cite{Fo84b}.
Subsequently, it was found that the general $(k_1,k_2)$-point correlation function could be expressed as a quaternion determinant (or equivalently a Pfaffian) \cite[Equation (6.168)]{forrester10}.
Well established theory relating to the correlation functions for (\ref{1.1}) with $\beta = 1$ and 4 \cite[\S 6]{forrester10} (see also Section 1.2 below) implies that there is a family of skew orthogonal polynomials underpinning this result. Appreciating this point, the authors of~ \cite{Rider12} considered a variant of (\ref{1.2}) on the real line specified by the
Boltzmann factor
\begin{align}\label{1.3}
\prod_{1\leq j<k\leq N_1}|\al_k-\al_j|\prod_{1\leq j<k\leq N_2}|\be_k-\be_j|^4\prod_{j=1}^{N_1}\prod_{k=1}^{N_2}|\al_j-\be_k|^2e^{-V(\al_j)-2V(\be_k)} 
\end{align}
with a Gaussian weight $V(x)=x^2/2$. It was found that the corresponding skew orthogonal polynomials can be expressed as a linear combination of Hermite polynomials with coefficients being particular Laguerre polynomials. Asymptotic results relating to the partition function and one-point correlation functions were given; see too the recent work \cite{borgo18}. This two-component log-gas plays a key role in the recent study \cite{feng19} relating to small gaps for the Gaussian orthogonal ensemble. Furthermore, excluding the weight factors, it is noted in \cite{chen06} that \eqref{1.3} results as the eigenvalue Jacobian for real symmetric matrices with a prescribed number of double degeneracies. In the present work, it is the skew orthogonal
polynomials which are our focus.

\subsection{Skew orthogonal polynomials arising from log-gases}
Since this paper is about the skew orthogonal polynomials arising from two-component log-gases (\ref{1.3}), 
it is appropriate to first revise how skew orthogonal polynomials appear in the corresponding one-component 
cases $N_1 = 0$ and $N_2=0$. In these cases, (\ref{1.3}) is proportional to the eigenvalue PDF for invariant ensembles of 
real symmetric matrices ($N_2=0$), and self dual Hermitian matrices $(N_1 = 0)$. With this interpretation,
the skew orthogonal polynomials most simply occur
 as the averaged characteristic polynomials \cite[\S 5]{metha04}. For $\beta=1$, the monic skew  orthogonal polynomials $\{p^{(1)}_j(x)\}_{j=0}^\infty$ are given by (up to a shift with arbitrary constant $c_n$)
\begin{align}\label{1.3a}
p^{(1)}_{2n}(x)=\langle \det(x-M)\rangle_{2n}^{(1)},\quad p^{(1)}_{2n+1}(x)=\langle (x+\text{Tr}M+c_n)\det(x-M)\rangle_{2n}^{(1)},
\end{align}
while for $\beta=4$, the monic skew  orthogonal polynomials $\{p^{(4)}_j(x)\}_{j=0}^\infty$ are given by
\begin{align}\label{1.3b}
p^{(4)}_{2n}(x)=\langle \det(x-M)\rangle_{n}^{(4)},\quad p^{(4)}_{2n+1}(x)=\langle (x+\text{Tr}M+c_n)\det(x-M)\rangle_{n}^{(4)}.
\end{align}
The label in these latter averages relates to the size of $M$ regarded as having quaternion entries. In practice each quaternion is represented as a particular $2 \times 2$ complex matrix
(see e.g.~\cite[\S 1.3.2]{forrester10}), so the sizes are consistent with the $\beta = 1$ case. The quantities in these averages depend only on the eigenvalues of $M$.
Consequently, they can each be written in terms of integrals over the corresponding eigenvalue PDFs. Thus,
\begin{align*}
p^{(1)}_{2n}(x) & = {1 \over C_{N,1}} \int_{I^n}  \prod_{1 \le j < k \le 2n} | x_k - x_j|\prod_{l=1}^{2n} (x - x_l)e^{-V(x_l)}dx_l, \\
p^{(1)}_{2n+1}(x)&= {1 \over C_{N,1}} \int_{I^n}  \prod_{1 \le j < k \le 2n} | x_k - x_j|(x+\sum_{l=1}^{2n} x_l+c_n)\prod_{l=1}^{2n} (x - x_l)e^{-V(x_l)}dx_l,\\
p^{(4)}_{2n}(x) & = {1 \over C_{N,4}} \int_{I^n}  \prod_{1 \le j < k \le n} ( x_k - x_j )^4\prod_{l=1}^{n} (x - x_l)^2 e^{-2V(x_l)}dx_l,\\
p^{(4)}_{2n+1}(x) & = {1 \over C_{N,4}} \int_{I^n}  \prod_{1 \le j < k \le n} ( x_k - x_j )^4(x+2\sum_{l=1}^{n}x_l+c_n)\prod_{l=1}^{n} (x - x_l)^2 e^{-2V(x_l)}dx_l,
\end{align*}
where $C_{N,\beta}$ denotes the normalisation, and $I$ the support of the spectrum.

It was demonstrated in \cite{dyson70} that the eigenvalue correlation functions of an invariant ensemble random matrix with $\beta=1$ or 4 can be written as a quaternion 
(or equivalently, Pfaffian), specified by a particular $2 \times 2$ correlation kernel.
Moreover, in \cite{mahoux91}, it was shown that the correlation kernel can be expressed in terms of the skew orthogonal polynomials.
Therefore, the problem of evaluating correlation kernels with $\beta=1,4$ is transformed into the evaluation of the skew orthogonal polynomials.
Much earlier it was recognised that the eigenvalue correlation  functions for an invariant ensemble with $\beta = 2$ can be written as a determinant specified by
a scalar correlation kernel expressible in terms of orthogonal polynomials (see e.g.~\cite[\S 5]{forrester10}).
Eventually this prompted the study of inter-relations between unitary matrix ensembles and orthogonal/symplectic matrix ensembles  \cite{adler00,nagao91,widom99, FR01}, and discrete analogues \cite{borodin09,Forrester19}.
For so called classical weight functions (see e.g.~\cite[\S 5.4.3]{forrester10}), it was demonstrated that the skew orthogonal polynomials can be written as a linear combination of corresponding
classical orthogonal polynomials. It was furthermore shown that the correlation kernel with $\beta=1,4$ can be written in terms of a rank one perturbation of a corresponding correlation kernel with $\beta=2$, the evaluation of the latter being well known from the Christoffel-Darboux summations. And the partition function, given generally as a product over the normalisation of
the skew orthogonal polynomials, permits in the classical cases evaluation in terms of
quantities relating to the orthogonal polynomial basis. Hence we are
motivated to find an expansion of the skew-orthogonal polynomials for
the two-component log-gas in terms of the underlying orthogonal polynomials.

\subsection{Main results}
The family of skew orthogonal polynomials corresponding to the Boltzmann factor \eqref{1.3} is skew orthogonal with respect to a particular skew symmetric inner product $\langle \cdot,\cdot\rangle_{X,\V}$, which is a linear combination of the $\langle\cdot,\cdot\rangle_{1,V_1}$ and $\langle\cdot,\cdot\rangle_{4,V_2}$. Since in \cite{Rider12} it has been shown that for the Gaussian weight, the skew orthogonal polynomials with $\langle \cdot,\cdot\rangle_{X,\V}$ can be written as a linear combination of Hermite polynomials with coefficients expressed by Laguerre polynomials (the latter have argument in terms of the fugacity---below a grand canonical formalism will be introduced), we pose the question as to whether all of the classical skew orthogonal polynomials under $\langle \cdot,\cdot\rangle_{X,\V}$ can be related to classical orthogonal polynomials. An affirmative answer is given in Proposition \ref{prop4}, showing that classical skew orthogonal polynomials under $\langle \cdot,\cdot\rangle_{X,\V}$ can be constructed from the classical orthogonal polynomials due to the Pearson equation. There are four cases to consider: the Gaussian, Laguerre, Jacobi and (generalised) Cauchy weight. We show that the hypergeometric orthogonal polynomials in the Askey-scheme appear when expressing the coefficients of the linear combination, extending the result of the Gaussian case for which the coefficients are given by  Laguerre polynomials (i.e. $_1F_1$ hypergeometric function). In particular, special attention is paid to the so-called generalised Cauchy weight in the form $(1-ix)^{c}(1+ix)^{\bar{c}}$ with $x\in\mathbb{R}$ and $c\in\mathbb{C}$. The model with the generalised Cauchy weight on the real line can be transformed into a generalised two-component log-gases with Jacobi weight on the unit circle under a stereographic projection, and therefore relates to the original work \cite{forrester84}.

\section{Skew orthogonal polynomials arising from a solvable mixed charge model}\label{sec2}
In this section, we present the general formalism showing how skew orthogonal polynomials enter into the computation of the partition function for the
mixed charge model (\ref{1.3}). As has been pointed
out above, the Gaussian weight case has previously been considered in \cite{Rider12}; our formalism works for general classical weights.

\subsection{A solvable mixed charge model}
Let us consider a statistical mechanical system of particles corresponding to two-dimensional charges (logarithmic potential) confined to a line.
Specifically, let there be $L$ particles of charge $+1$ and $M$ particles of charge $+2$, so that the total charge is $L+2M=K$. For convenience, we will assume that the total charge $K=2N$ is even, which implies $L$ is even too. The configuration space $\xi$ of this 1-dimensional system is a subspace of $\mathbb{R}^L\times \mathbb{R}^M$, defined as $\xi=\{(\xi_1;\xi_2)=(\alpha_1,\cdots,\alpha_L;\beta_1,\cdots,\beta_M),\, \xi_1\in\mathbb{R}^L,\,\xi_2\in\mathbb{R}^M\}$, where $\{\alpha_i\}_{i=1}^L$ and $\{\beta_j\}_{j=1}^M$ represent the locations of the charge $+1$ and $+2$ particles respectively. The particles are each to subject to an external confining potential.
The total energy of this model is given by
\begin{align*}
E(\xi_1;\xi_2)&=
\sum_{1\leq j<k\leq L}\log|\al_k-\al_j|+4\sum_{1\leq j<k\leq M}\log|\be_k-\be_j|\\&+2\sum_{j=1}^L\sum_{k=1}^M\log|\al_j-\be_k|
-\sum_{j=1}^L V_1(\al_j)-2\sum_{k=1}^M V_2(\be_k),
\end{align*}
where the first three terms correspond to the potential energy of the state $\xi=(\xi_1;\xi_2)$ and the last two are due to the external confining potentials $V_1,V_2:\,\mathbb{R}\to[0,\infty)$. Note that the potential functions of the two 
species are not assumed to be the same (cf.~\cite{Rider12}).

The partition function $Z_{L,M}$ of a fixed pair of species numbers $(L,M)$ is specified by
\begin{align*}
Z_{L,M}=\frac{1}{L!M!}\int_{\mathbb{R}^L\times\mathbb{R}^M}e^{-E(\xi_1;\xi_2)}d\mu^L(\al)d\mu^M(\be)
\end{align*}
with Boltzmann factor 
\begin{align}
\begin{aligned}\label{vandermonde}
e^{-E(\xi_1;\xi_2)}&=\prod_{1\leq j<k\leq L}|\al_k-\al_j|\prod_{1\leq j<k\leq M}|\be_k-\be_j|^4\prod_{j=1}^L\prod_{k=1}^M|\al_j-\be_k|^2 e^{-V_1(\al_j)-2V_2(\be_k)}\\
&=\left|\det\left(\alpha_j^ke^{-V_1(\alpha_j)}\quad \beta_{j'}^ke^{-V_2(\be_{j'})}\quad k\beta_{j'}^{k-1}e^{-V_2(\be_{j'})}\right)\right|_{\mbox{\tiny$\begin{array}{c}
j=1,\cdots,L,\\
j'=1,\cdots,M,\\
k=0,\cdots,2N-1.\end{array}$}}
\end{aligned}
\end{align}
The second line in (\ref{vandermonde}) follows from a confluent form of the Vandermonde determinant identity; see \cite{forrester84}. 

To most clearly exhibit solvable structures, we do not consider the partition function directly, but rather first form 
a particular  grand canonical ensemble, obtained by simultaneously considering all systems with species numbers
$L$ and $M$ such that $L+2M=2N$. To distinguish the species number making up the grand canonical ensemble,
we introduce a fugacity (generating function parameter) $X$ in the specification of the corresponding
grand canonical partition function,
\begin{align*}
Z_N(X)=\sum_{\adm}X^LZ_{L,M}=\sum_{\adm}\frac{X^L}{L!M!}\int_{\mathbb{R}^L\times\mathbb{R}^M}e^{-E(\xi_1;\xi_2)}d\mu^L(\al)d\mu^M(\be),
\end{align*}
where here  $\adm=\{(L,M)\,|\, L+2M=2N,\,L,M\in\mathbb{Z}_{\geq0}\}$.

First we will demonstrate that the partition function $Z_N(X)$ has a Pfaffian representation.

\begin{proposition}\label{prop1}
The partition function $Z_N(X)$ can be written as a Pfaffian 
\begin{align}\label{partition}Z_N(X)=2^N\Pf(m_{i,j}(X))_{i,j=0}^{2N-1},
\end{align} 
where the moments $m_{i,j}(X)$ are defined as
\begin{align}\label{skewinnerproduct}
m_{i,j}(X)=\frac{X^2}{2}\int_{\mathbb{R}\times \mathbb{R}}y^iz^j\text{sgn}(z-y)e^{-V_1(y)-V_1(z)}d\mu(y)d\mu(z)+\frac{j-i
}{2}\int_{\mathbb{R}}y^{i+j-1}e^{-2V_2(y)}d\mu(y).
\end{align}
\end{proposition}
\begin{proof}
We begin by performing a Laplace expansion (see e.g.~\cite[\S 3.3]{vein}) of the confluent Vandermonde determinant in \eqref{vandermonde} to express the determinant as sums in terms of determinants of species $\al$ and $\be$ respectively, 
\begin{align*}
e^{-E(\xi_1;\xi_2)}&=\sum_{\sigma}\text{sgn}\sigma\cdot \left|\det\left(\al_j^{\sigma(\eta)}e^{-V_1(\al_j)}\right)_{\mbox{\tiny$\begin{array}{c}
j,\eta=1,\cdots,L.
\end{array}$}}\right|\\
&\times\left|\det\left(\beta_{j'}^{\sigma'(\eta')}e^{-V_2(\be_{j'})},\,\sigma'(\eta')\beta_{j'}^{\sigma'(\eta')-1}e^{-V_2(\be_{j'})}\right)_{\mbox{\tiny$\begin{array}{c}
j'=1,\cdots,M,\\
\eta'=1,\cdots,2M
\end{array}$}}\right|,
\end{align*}
where $\sigma$ is injective map from $\{1,\,\cdots,\,L\}\mapsto\{1,\,\cdots,\,2N\}$ satisfying $\sigma(1)<\cdots<\sigma(L)$ and the image of $\sigma'$ is $\text{Im}(\sigma')=\{1,\,\cdots,\,2N\}\backslash \text{Im}(\sigma)$ satisfying $\sigma'(1)<\cdots<\sigma'(2M)$. Moreover, the sum $\sum_\sigma$ represents the sum of all the possibilities of $L$ integer partitions, contributing a factor $C_{2N}^L$. 

Since the determinant has been expanded into terms involving species $\al$ and $\be$ independently, we can similarly factor the integrations
\begin{align*}
\frac{X^L}{L!M!}\int_{\mathbb{R}^L\times\mathbb{R}^M} e^{-E(\xi_1;\xi_2)}d\mu^L(\al)&d\mu^M(\be)=\sum_\sigma \text{sgn}\sigma\cdot \frac{X^L}{L!}\int_{\mathbb{R}^L}\left|\det\left(\al_j^{\sigma(j)}e^{-V_1(\al_j)}\right)\right|d\mu^L(\alpha)\\&\times \frac{1}{M!}\int_{\mathbb{R}^M}\left|\det\left(\beta_{j'}^{\sigma'(j')}e^{-V_2(\be_{j'})},\,\sigma'(j')\beta_{j'}^{\sigma'(j')-1}e^{-V_2(\be_{j'})}\right)\right|d\mu^M(\be).
\end{align*}
By making use of integration formulas due to de Bruijn \cite{dB55} for both factors,
the right hand side of the above equation is equal to
\begin{align*}
\sumsigma\cdot\Pf(X^2\tilde{m}_{\sigma(i),\sigma(j)})_{i,j=1}^L \Pf(\hat{m}_{\sigma'(i),\sigma'(j)})_{i,j=1}^{2M},
\end{align*}
where
\begin{align*}
&\tilde{m}_{i,j}=\int_{\mathbb{R}\times\mathbb{R}}y^{i-1}z^{j-1}\text{sgn}(z-y)e^{-V_1(y)-V_1(z)} \, d\mu(y)d\mu(z),\\
&\hat{m}_{i,j}=(j-i)\int_{\mathbb{R}}y^{i+j-3} e^{-2V_2(y)} \, d\mu(y).
\end{align*}
We can now make use of  \cite[Lemma 4.2]{stembridge90} to sum the Pfaffian according to
\begin{align*}
\sum_{\adm}\sumsigma\cdot\Pf(X^2\tilde{m}_{\sigma(i),\sigma(j)})_{i,j=1}^L \Pf(\hat{m}_{\sigma'(i),\sigma'(j)})_{i,j=1}^{2M}=\Pf(X^2\tilde{m}_{i,j}+\hat{m}_{i,j})_{i,j=1}^{2N}.
\end{align*}
By denoting $2m_{i+1,j+1}(X)=X^2 \tilde{m}_{i,j}+\hat{m}_{i,j}$, we complete the proof.
\end{proof}

\subsection{Skew orthogonal polynomials arising from the solvable mixed charge model}
In keeping with Proposition \ref{prop1}, we introduce the skew symmetric inner product 
\begin{align*}
\langle \cdot,\cdot\rangle_{X,\V}: \quad \mathbb{C}[z]\times\mathbb{C}[z]&\to\mathbb{C}\\
(z^i,\,z^j)&\mapsto m_{i,j}(X)=X^2\langle z^i, z^j\rangle_{1,V_1}+\langle z^i,z^j\rangle_{4,V_2}.
 \end{align*}
Here, with weight function $V_1,\,V_2$ specified in \eqref{skewinnerproduct}, the skew inner products $\langle \cdot,\cdot\rangle_\beta$ with $\beta=1,2,4$ are defined by
\begin{align*}
\langle \phi(z),\psi(z)\rangle_{1,\,V}&=\frac{1}{2}\int_{\mathbb{R}\times\mathbb{R}}\phi(y)\psi(z)\text{sgn}(z-y)e^{-V(y)-V(z)}d\mu(y)d\mu(z),\\
\langle \phi(z),\psi(z)\rangle_{2,\,V}&=\int_{\mathbb{R}}\phi(z)\psi(z)e^{-2V(z)}d\mu(z),\\
\langle \phi(z),\psi(z)\rangle_{4,\,V}&=\frac{1}{2}\int_{\mathbb{R}}(\phi(z)\partial_z\psi(z)-\partial_z\phi(z)\psi(z))e^{-2V(z)}d\mu(z).
\end{align*}
 Following the procedure of skew symmetric tri-diagonalisation \cite{adler99,Forrester19}, we know that the skew symmetric moment matrix $M:=\left(m_{i,j}(X)\right)_{i,j=0}^{2N-1}$ can be decomposed into $M=S^{-1}JS^{-\top}$, where $S$ is a lower triangular matrix with diagonals $1$ and $J$ is a block matrix of the form
\begin{align}\label{sbd}
J=\text{diag}\left\{
\left(\begin{array}{cc}
0&u_0\\
-u_0&0
\end{array}
\right),
\cdots,
\left(\begin{array}{cc}
0&u_{N-1}\\
-u_{N-1}&0
\end{array}
\right)
\right\},\,u_n=\frac{\tau_{2n+2}}{\tau_{2n}},\,
\tau_{2n}=\Pf(m_{i,j})_{i,j=0}^{2n-1}
\end{align}
with $\tau_0=1$. Due to the decomposition, we can introduce a family of polynomials $\{Q_n(z)\}_{n=0}^{2N-1}$ such that
\begin{align*}
Q_n(z)=(S\chi(z))_{n+1},\quad \chi(z)=(z^0,z^1,\cdots,z^{2N-1})^\top,
\end{align*}
where $(S\chi(z))_{n+1}$ denotes the $(n+1)$-th component of vector $(S\chi(z))$. These polynomials skew diagonalise $\langle\cdot,\cdot\rangle_{X,\V}$.
\begin{proposition}
The family of polynomials $\{Q_n(z)\}_{n=0}^{2N-1}$ are skew orthogonal with respect to the inner product $\langle\cdot,\cdot\rangle_{X,\V}$. In other words, we have
\begin{align}\label{sop}
\begin{aligned}
&\langle Q_{2n}(z),Q_{2m+1}(z)\rangle_{X,\V}=u_n\delta_{n,m},\\
&\langle Q_{2m}(z),Q_{2n}(z)\rangle_{X,\V}=\langle Q_{2m+1}(z),Q_{2n+1}(z)\rangle_{X,\V}=0 \end{aligned}
\end{align} 
for $n,\,m=0,1,\cdots,N-1$.
\end{proposition}
\begin{proof}
By denoting $Q(z)=(Q_0(z),\cdots,Q_{2N-1}(z))^\top=S\chi(z)$, one can easily find $$\langle Q(z), Q^\top(z)\rangle_{X,\V}=S\langle \chi(z),\chi^\top(z)\rangle_{X,\V}S^\top=SMS^\top=J.$$
\end{proof}
\begin{remark}
Since the skew orthogonal relation \eqref{sop} is equivalent to a linear system, one can obtain a Pfaffian form of the skew orthogonal polynomials following the procedure in \cite[Section 2]{chang18}. Since we mainly consider the relation between the skew orthogonal polynomials and orthogonal polynomials in this paper, we omit the Pfaffian forms here.
\end{remark}

The skew symmetric inner product defined in \eqref{skewinnerproduct} is a linear combination of skew symmetric inner products induced by random matrices with orthogonal and symplectic invariance (i.e. $\beta=1,4$ respectively). As is known from \cite{adler00} , the skew orthogonal polynomials with $\beta=1,4$ can be expressed as a linear combination of orthogonal polynomials in the classical cases (a precise meaning of `classical' is
given in the paragraph above Proposition \ref{prop4}).  Therefore, we seek a relationship between the skew orthogonal polynomials arising from \eqref{skewinnerproduct} and orthogonal polynomials.
Important for this purpose is the following proposition.

\begin{proposition}\label{prop3}
For specific coefficients $\{\zeta_j\}_{j=0}^\infty$, if $\{p_j(z)\}_{j=0}^\infty$ are monic polynomials satisfying
\begin{align}\label{rel1}
\begin{aligned}
\left\{\begin{array}{l}
\langle p_{2m}(z),p_{2n+1}(z)-\zeta_n p_{2n-1}(z)\rangle_{1,V_1}=\ho_n\delta_{n,m},\\
\langle p_{2m}(z),p_{2n}(z)\rangle_{1,V_1}=\langle p_{2n+1}(z)-\zeta_n p_{2n-1}(z),p_{2m+1}(z)-\zeta_m p_{2m-1}(z)\rangle_{1,V_1}=0,\\
\langle p_{2m}(z),p_{2n+1}(z)\rangle_{4,V_2}=\hf_{2n}\delta_{n,m}-\hf_{2m-1}\delta_{m-1,n},\\
\langle p_{2m}(z),p_{2n}(z)\rangle_{4,V_2}=\langle p_{2m+1}(z),p_{2n+1}(z)\rangle_{4,V_2}=0,
\end{array}
\right.
\end{aligned}
\end{align}
then there exist a family of monic skew orthogonal polynomials $\{Q_j(z)\}_{j=0}^\infty$ 
\begin{align}\label{new}
Q_{2m}(z)=\frac{1}{\al_m}\sum_{k=0}^m \al_k p_{2k}(z),\quad Q_{2m+1}(z)=\frac{1}{\xi_m}\sum_{k=0}^m \xi_k(p_{2k+1}(z)-\zeta_k p_{2k-1}(z)),
\end{align}
skew orthogonal under the skew symmetric inner product $\langle\cdot,\cdot\rangle_{X,\V}=X^2\langle \cdot,\cdot\rangle_{1,V_1}+\langle \cdot,\cdot\rangle_{4,V_2}$. Moreover, the coefficients $\{\al_j\}_{j=0}^\infty$ and $\{\xi_j\}_{j=0}^\infty$ satisfy the  recurrence relations
\begin{subequations}
\begin{align}
&\hf_{2j+1}\al_{j+1}-\left(X^2\ho_j+\hf_{2j}+\zeta_j\hf_{2j-1}\right)\al_j+\zeta_j\hf_{2j-2}\al_{j-1}=0\label{ra}\\
&\zeta_{j+1}\hf_{2j}\xi_{j+1}-\left(X^2\ho_j+\hf_{2j}+\zeta_j\hf_{2j-1}\right)\xi_j+\hf_{2j-1}\xi_{j-1}=0\label{rx}
\end{align}
\end{subequations}
with initial values $\al_{-1}=\xi_{-1}=0$ and $\al_0=\xi_0=1$.
\end{proposition}

\begin{proof}
Since $\{p_{2j}(z),\, p_{2j+1}(z)-\zeta_jp_{2j-1}(z)\}_{j=0}^\infty$ form a basis of $L^2(d\mu)$, we can write the monic polynomials $\{Q_{j}(z)\}_{j=0}^\infty$ as linear combinations of this basis
\begin{subequations}
\begin{align}
&Q_{2m}(z)=\sum_{k=0}^m \alm_kp_{2k}(z)+\sum_{k=0}^{m-1}\bem_k(p_{2k+1}(z)-\zeta_kp_{2k-1}(z)),\quad \alm_m=1,\label{6.eq1}\\
&Q_{2m+1}(z)=\sum_{k=0}^m \xim_k(p_{2k+1}(z)-\zeta_kp_{2k-1}(z))+\sum_{k=0}^m \etam_kp_{2k}(z),\quad \xim_m=1. \label{6.eq2}
\end{align}
\end{subequations}
The skew orthogonal relations \eqref{sop} can be written as \footnote{For the third condition, the skew orthogonal relation only hold for $j<m$ since we haven't assume $Q_{2m+1}(z)$ is independent of $Q_{2m}(z)$ yet.}
\begin{subequations}
\begin{align}
&\langle Q_{2m}(z),p_{2j+1}(z)-\zeta_jp_{2j-1}(z)\rangle_{X,\V}=0,\quad \text{for $j<m$};\label{eq1}\\
&\langle Q_{2m}(z),p_{2j}(z)\rangle_{X,\V}=0,\quad \text{for $j\leq m$};\label{eq2}\\
&\langle Q_{2m+1}(z),p_{2j+1}(z)-\zeta_jp_{2j-1}(z)\rangle_{X,\V}=0,\quad \text{for $j<m$};\label{eq3}\\
&\langle Q_{2m+1}(z),p_{2j}(z)\rangle_{X,\V}=0,\quad \text{for $j<m$}.\label{eq4}
\end{align}
\end{subequations}
From (\ref{6.eq1}) and (\ref{6.eq2}), the equation \eqref{eq1} is equivalent to
\begin{align*}
\sum_{k=0}^m\alm_k\langle p_{2k}(z),p_{2j+1}(z)-\zeta_jp_{2j-1}(z)\rangle_{X,\V}=0,\quad \text{for $j<m$}
\end{align*}
since $\langle p_{2k+1}(z)-\zeta_kp_{2k-1}(z),p_{2j+1}(z)-\zeta_jp_{2j-1}(z)\rangle_{X,V}=0$ for arbitrary $j,k\in\mathbb{N}$ according to \eqref{rel1}. Making further use of \eqref{rel1} shows
\begin{align*}
X^2\sum_{k=0}^m\alm_k\ho_k\delta_{k,j}+\sum_{k=0}^m \alm_k(\hf_{2j}\delta_{k,j}-\hf_{2k-1}\delta_{k-1,j})-\sum_{k=0}^m\alm_k\zeta_j(\hf_{2j-2}\delta_{k,j-1}-\hf_{2k-1}\delta_{k,j})=0,
\end{align*}
and thus
\begin{align}\label{a8}
\hf_{2j+1}\alm_{j+1}-\left(X^2\ho_j+\hf_{2j}+\zeta_j\hf_{2j-1}\right)\alm_j+\zeta_j\hf_{2j-2}\alm_{j-1}=0,\quad\text{for $j<m$}.
\end{align}
In fact, the above equation is independent of $m$ and the solution is determined by the boundary values $\alm_{-1}=0$ and $\alm_m=1$. An efficient way is to set $\alm_j=\al_j/\al_m$, showing all of the coefficients satisfy the same recurrence relation \eqref{ra} with initial data 
\begin{equation}\label{da1}
\al_{-1}=0, \qquad \al_0=1
\end{equation}
 and 
\begin{align*}
Q_{2m}(z)=\frac{1}{\al_m}\sum_{k=0}^m \al_k p_{2k}(z)+\sum_{k=0}^{m-1}\bem_k(p_{2k+1}(z)-\zeta_kp_{2k-1}(z)).
\end{align*}
Now substituting the above form of $Q_{2m}(z)$ into \eqref{eq2}, one can see 
\begin{align*}
\sum_{k=0}^{m-1}\bem_k\langle p_{2k+1}(z)-\zeta_kp_{2k-1}(z),p_{2j}(z)\rangle_{X,V}=0,
\end{align*}
which leads us to
\begin{align*}
-X^2\sum_{k=0}^{m-1}\bem_k\ho_k\delta_{k,j}&-\sum_{k=0}^{m-1}\bem_k(\hf_{2k}\delta_{k,j}-\hf_{2j-1}\delta_{j-1,k})\\
&+\sum_{k=0}^{m-1}\bem_k\zeta_k(\hf_{2k-2}\delta_{k-1,j}-\hf_{2j-1}\delta_{k,j})=0,\quad \text{for $j\leq m$}.
\end{align*}
Taking $j$ from $m$ to $0$ iteratively we find $\bem_{m-1}=\cdots=\bem_0=0$. Therefore, from \eqref{eq1} and \eqref{eq2}, we obtain
\begin{align*}
Q_{2m}(z)=\frac{1}{\al_m}\sum_{k=0}^m\al_kp_{2k}(z),
\end{align*}
where $\{\al_k\}$ satisfy \eqref{ra} with initial values $\al_{-1}=0$ and $\al_0=1$.

Now we turn to equations \eqref{eq3} and \eqref{eq4}. First consider  \eqref{eq3}.
From (\ref{6.eq1}) and (\ref{6.eq2}), this is equivalent to
\begin{align*}
\sum_{k=0}^m \etam_k\langle p_{2k}(z),p_{2j+1}(z)-\zeta_jp_{2j-1}(z)\rangle_{X,\V}=0,\quad \text{for $j<m$}.
\end{align*}
Again, by the use of conditions \eqref{rel1}, one can show that this implies
\begin{align*}
\hf_{2j+1}\etam_{j+1}-\left(X^2\ho_j+\hf_{2j}+\zeta_j\hf_{2j-1}\right)\etam_j+\zeta_j\hf_{2j-2}\etam_{j-1}=0,\quad \text{for $j<m$},
\end{align*}
demonstrating that $\{\etam_j\}$ admit the same recurrence relation with $\{\alm_j\}$ and therefore
\begin{align*}
Q_{2m+1}(z)=\sum_{k=0}^m\xim_k(p_{2k+1}(z)-\zeta_kp_{2k-1}(z))+c_mQ_{2m}(z)
\end{align*}
with a constant $c_m$. A basic property of skew orthogonal polynomials, is that the transformation $Q_{2m+1}(z) \mapsto Q_{2m+1}(z)+c_mQ_{2m}(z)$ does not change the skew orthogonality
(this explains the arbitrary constant $c_n$ in (\ref{1.3a}) and (\ref{1.3b})), and therefore we are free to set $c_m = 0$, giving
\begin{align*}
Q_{2m+1}(z)=\sum_{k=0}^m \xim_k(p_{2k+1}(z)-\zeta_kp_{2k-1}(z)).
\end{align*}
We now make use of \eqref{eq4} by substituting this form to obtain
\begin{align*}
-X^2\sum_{k=0}^m \xim_k\ho_k\delta_{k,j}-\sum_{k=0}^m\xim_k(\hf_{2k}\delta_{k,j}-\hf_{2j-1}\delta_{k,j-1})+\sum_{k=0}^m \xim_k\zeta_k(\hf_{2k-2}\delta_{k-1,j}-\hf_{2j-1}\delta_{k,j})=0
\end{align*}
\text{for $j<m$}. Analogous to the analysis of (\ref{a8}) we set  $\xim_k=\xi_k/\xi_m$ to conclude
\begin{align*}
Q_{2m+1}(z)=\frac{1}{\xi_m}\sum_{k=0}^m \xi_k(p_{2k+1}(z)-\zeta_kp_{2k-1}(z)),
\end{align*}
where $\{\xi_k\}$ satisfy \eqref{rx} with initial values $\xi_{-1}=0$ and $\xi_0=1$.
\end{proof}

The coupling between the coefficients $\{\al_j\}$ and $\{\xi_j\}$ implied by the recurrences \eqref{ra} and \eqref{rx}
can be made more explicit.
\begin{corollary}\label{cjj}
We have
\begin{align}\label{cj}
 \xi_j=c_{j}^{-1}\al_j, \quad \text{where}\quad c_j=\prod_{k=1}^j\zeta_k\frac{\hf_{2k-2}}{\hf_{2k-1}}.
\end{align}
\end{corollary}
\begin{proof}
Note that if we make the transformation of variables $\xi_j=c_j^{-1}\hat{\xi}_j$,  the recurrence relation \eqref{rx} reads
\begin{align*}
\zeta_{j+1}\hf_{2j}\frac{c_{j}}{c_{j+1}}\hat{\xi}_{j+1}-\left(X^2\ho_j+\hf_{2j}+\zeta_j\hf_{2j-1}\right)\hat{\xi}_j+\hf_{2j-1}\frac{c_{j}}{c_{j-1}}\hat{\xi}_{j-1}=0.
\end{align*}
To be consistent with \eqref{ra}, it is further required
\begin{align*}
\zeta_{j+1}\hf_{2j}\frac{c_j}{c_{j+1}}=\hf_{2j+1},\quad \hf_{2j-1}\frac{c_j}{c_{j-1}}=\zeta_j\hf_{2j-2}.
\end{align*}
Hence
\begin{align*}
c_j= C \prod_{k=1}^j\zeta_k\frac{\hf_{2k-2}}{\hf_{2k-1}},\quad \text{and} \quad \xi_j=c_j^{-1}\al_j,
\end{align*}
for some constant $C$.
By noting $\al_0=\xi_0=1$, it follows that $C = 1$, thus completing  the proof.
\end{proof}

The condition of proposition \eqref{prop3} is very general. Thus, starting with two arbitrary  skew symmetric inner products $\langle \cdot, \cdot\rangle_{s_1}$ and $\langle \cdot,\cdot\rangle_{s_2}$, if a family of (not necessarily orthogonal)  polynomials $\{p_j(z)\}_{j=0}^\infty$ satisfy the condition \eqref{rel1}, then we can define a new skew symmetric inner product $\langle \cdot, \cdot\rangle_{new}=X^2\langle\cdot,\cdot\rangle_{s_1}+\langle \cdot,\cdot\rangle_{s_2}$ with corresponding skew orthogonal polynomials specified by \eqref{new}, \eqref{ra} and \eqref{rx}. 

The task then is to specify families of skew symmetric inner products and polynomials $\{p_j(z)\}_{j=0}^\infty$ satisfying \eqref{rel1}.
In fact we find that all of the classical monic orthogonal polynomials $\{p_j(x)\}_{j=0}^\infty$ with weight $e^{-2V}$, characterised by the Pearson equation $2V'(z)=g(z)/f(z)$ with $\text{deg}~ g\leq 1$ and $\text{deg}~ f\leq 2$, satisfy the condition \eqref{rel1} with
\begin{equation}\label{214}
V_1=V+\frac{1}{2}\log f, \qquad V_2=V-\frac{1}{2}\log f. 
\end{equation}
It is remarkable that for the classical weights, the procedure to find the relations between $\beta=1,\,4$ and $\beta=2$ is to construct a linear operator \cite{adler00}
\begin{align*}
\mathcal{A}=f\partial_z+\frac{f'-g}{2}
\end{align*}
such that $\langle \phi, \mathcal{A}^{-1}\psi\rangle_{2,V}=-\langle \phi,\psi\rangle_{1,V_1}$ and $\langle \phi,\mathcal{A}\psi\rangle_{2,V}=\langle \phi,\psi\rangle_{4,V_2}$. In this case, the inner product $\langle\cdot,\cdot\rangle_{X,V_1,V_2}$ can be constructed from the operator $X^2 \mathcal{A}^{-1}+\mathcal{A}$, but it needs more computations which are not considered in this paper.

\begin{proposition}\label{prop4}
All classical orthogonal polynomials, in the sense of the requirement on the Pearson equation as specified above, satisfy \eqref{rel1}, where the relations (\ref{214}) are assumed.
\end{proposition}
\begin{proof}
The proof of this claim is based on the Pearson equation. The first two relations have been demonstrated in \cite{adler00}---the key idea is to use the Pearson pair $(f,g)$ to find the skew orthogonal polynomials with $\beta=1$; a detailed procedure can be found in the Appendix \ref{sop1}. Moreover, the quantities $\{\zeta_j\}_{j=0}^\infty$ and $\{\ho_j\}_{j=0}^\infty$ are computed from the procedure. The last two relations can also be obtained from the Pearson pair. In \cite{alsalam72}, it has been demonstrated that the classical orthogonal polynomials $\{p_n(z)\}_{n=0}^\infty$ which are characterised by the Pearson pair $(f,g)$ must satisfy
\begin{align*}
f(z)p_n'(z)=a_np_{n+1}(z)+b_np_n(z)+c_np_{n-1}(z),
\end{align*}
where $\{a_n\}_{n=0}^\infty$, $\{b_n\}_{n=0}^\infty$, $\{c_n\}_{n=0}^\infty$ are three sequences independent of $z$. Therefore, one can obtain
\begin{align*}
\langle p_{2m}(z),p_{2n}(z)\rangle_{4,V_2}&=\langle p_{2m+1}(z),p_{2n+1}(z)\rangle_{4,V_2}=0,\\
\langle p_{2m}(z),p_{2n+1}(z)\rangle_{4, V_2}&=\frac{1}{2}\langle p_{2m}(z),f(z)p'_{2n+1}(z)\rangle_{2,V}-\frac{1}{2}\langle f(z)p'_{2m}(z),p_{2n+1}(z)\rangle_{2, V}\\
&=\frac{1}{2}(c_{2n+1}h_{2n}-a_{2n}h_{2n+1})\delta_{n,m}-\frac{1}{2}(c_{2m}h_{2m-1}-a_{2m-1}h_{2m})\delta_{m,n+1},
\end{align*}
where $h_n=\langle p_n(z),p_n(z)\rangle_{2,V}$. We complete the proof by defining $\hf_j=\frac{1}{2}(c_{j+1}h_j-a_jh_{j+1})$.
\end{proof}

In conclusion to this section, we demonstrate that the partition function $Z_N(X)$ in \eqref{partition} can be calculated by use of the skew orthogonal polynomials constructed in Proposition \ref{prop3}, thus
exhibiting the solvability of the model. To begin, note that
\begin{align*}
Z_N(X)\overset{\eqref{partition}}{=}2^N\tau_{2N}\overset{\eqref{sbd}}{=}2^N\prod_{j=0}^{N-1}u_j\overset{\eqref{sop}}{=}2^N\prod_{j=1}^N\langle Q_{2j}, Q_{2j+1}\rangle_{X,\V}.
\end{align*}
Substituting the skew orthogonal polynomials $\{Q_j(z)\}_{j=0}^\infty$ in terms of orthogonal polynomials via \eqref{new} shows
\begin{align*}
2^N&\prod_{j=0}^{N-1}\langle Q_{2j},Q_{2j+1}\rangle_{X,\V}=2^N\prod_{j=0}^{N-1}\frac{1}{\al_j\xi_j}\sum_{k_1=0}^j\sum_{k_2=0}^j \al_{k_1}\xi_{k_2}\langle p_{2k_1},p_{2k_2+1}-\zeta_{k_2}p_{2k_2-1}\rangle_{X,\V}.
\end{align*}
Making the use of \eqref{rel1}, we see this is equal to
\begin{align*}
2^N\prod_{j=0}^{N-1} \frac{1}{\al_j\xi_j}\left(
\sum_{k=0}^j \al_k\xi_k(X^2\ho_k+\hf_{2k}+\zeta_k\hf_{2k-1})-\sum_{k=1}^j \al_{k-1}\xi_k\zeta_k\hf_{2k-2}-\sum_{k=1}^j\al_k\xi_{k-1}\hf_{2k-1}
\right).
\end{align*}
By taking into account the recurrence relation \eqref{ra}, we conclude
\begin{align*}
Z_N(X)=2^N \al_N \prod_{j=0}^{N-1}\hf_{2j+1}.
\end{align*}

In the classical cases, $\{\hf_{j}\}_{j=0}^\infty$ are known normalisations, expressible in terms of quantities relating to the corresponding orthogonal polynomials \cite{adler00}. We will show below that in these cases $\al_N$, and thus the partition function, can be written in terms of 
hypergeometric orthogonal polynomials. Interestingly, this class of polynomials has appeared in a recent work in the moments of the spectral density for the classical weights with $\beta=2$ \cite{cunden2019}. 

\section{The classical cases}

We know from Proposition \ref{prop4} that the classical weights allow for a structured form of the 
skew orthogonal polynomials generated in terms of the classical orthogonal polynomials.
Here we will show that the coefficients $\{\al_j\}_{j=0}^\infty$ and $\{\xi_j\}_{j=0}^\infty$ can be expressed in terms of hypergeometric orthogonal polynomials from the Askey-Wilson table. 
There are four classical weights: the Gaussian, Laguerre, Jacobi and generalised Cauchy. These will be considered in turn.

\subsection{Gaussian case}
The two-component plasma for mixed charges with Gaussian weight has already been considered in \cite{Rider12}. We restate the main results for the completeness of our work. It is a standard result that the monic Hermite polynomials $\{p_k(z)\}_{k=0}^\infty$ satisfy the orthogonal relation
\begin{align*}
\langle p_k(z), p_j(z)\rangle_{2,V}=2^{-k}\sqrt{\pi}\Gamma(k+1)\delta_{k,j},
\end{align*}
where the weight function is given by $e^{-2V(z)}=e^{-z^2}$. From the weight function, we can compute the Pearson pair as $(f,g)=(1,2z)$, and therefore $e^{-V_1(z)}=e^{-V_2(z)}=z^2/2$, meaning that the external potentials of these two different particles are exactly the same.
Moreover, from \cite[Equation 3.2]{adler00}, we have
\begin{align*}
\langle p_{2m},p_{2n+1}-np_{2n-1}\rangle_{1,V_1}=2^{-2n}\sqrt{\pi}\Gamma(2n+1)\delta_{n,m},
\end{align*}
and from the derivative formula $\partial_zp_k(z)=kp_{k-1}(z)$, one knows
\begin{align*}
\langle p_{2m}, p_{2n+1}\rangle_{4,V_2}&=\frac{1}{2} (2n+1)\langle p_{2m},p_{2n}\rangle_{2,V_2}-\frac{1}{2} (2m)\langle p_{2m-1},p_{2n+1}\rangle_{2,V_2}\\
&=2^{-2n-1}\sqrt{\pi}\Gamma(2n+2)\delta_{n,m}-2^{-2m}\sqrt{\pi}\Gamma(2m+1)\delta_{m-1,n}.
\end{align*}
Therefore, the coefficients of the Hermite polynomials with regard to the  relations \eqref{rel1} are
\begin{align}\label{gauss}
\zeta_j=j, \quad \ho_j=2^{-2j}\sqrt{\pi}\Gamma(2j+1),\quad \hf_{j}=2^{-j-1}\sqrt{\pi}\Gamma(j+2).
\end{align}
Substituting them into the recurrence relation \eqref{ra}, one finds that the coefficients $\{\al_j\}$ satisfy
\begin{align*}
(2j+2)(2j+1)\al_{j+1}-2(2X^2+4j+1)\al_j+4\al_{j-1}=0,
\end{align*}
which, upon scaling $\al_j=2^{2j}\frac{j!}{(2j)!}\hat{\al}_j$, reduces to
\begin{align*}
(j+1)\hat{\al}_{j+1}-({X^2}+2j+\frac{1}{2})\hat{\al}_j+(j-\frac{1}{2})\hat{\al}_{j-1}=0.
\end{align*}
This equation, in combination with the initial conditions (\ref{da1}),
 can be recognised as the three term recurrence satisfied by the
Laguerre polynomials with parameter $-1/2$, and so $\hat{\al}_j=L_j^{-1/2}(-X^2)$, where $\{L_n^a(z)\}_{n=0}^\infty$ denotes the standard (not monic!) Laguerre polynomials with parameter $a$ through the orthogonal relation
\begin{align*}
\int_0^\infty L_n^a(z)L_m^a(z)z^{a}e^{-z}dx=\frac{\Gamma(a+n+1)}{\Gamma(n+1)}\delta_{n,m}.
\end{align*}
Therefore, we can conclude that
\begin{align*}
\al_j=2^{2j}\frac{j!}{(2j)!}L_j^{-\frac{1}{2}}(-{X^2}).
\end{align*}
With regard of the coefficients $\{\xi_j\}_{j=0}^\infty$, since
\begin{align*}
\zeta_{j+1}\hf_{2j}=(j+1) 2^{-2j-1}\sqrt{\pi}\Gamma(2j+2)=2^{-2j-2}\sqrt{\pi}\Gamma(2j+3)=\hf_{2j+1},
\end{align*}
one sees that $\{\al_j\}_{j=0}^\infty$ and $\{\xi_j\}_{j=0}^\infty$ satisfy the same recurrence relation with same initial values. Therefore,
\begin{align*}
\al_j=\xi_j=2^{2j}\frac{j!}{(2j!)}L_j^{-\frac{1}{2}}(-X^2),
\end{align*} 
as demonstrated in \cite[Section 4]{Rider12}.

\subsection{Laguerre case}
The Laguerre case is to consider the weight function $e^{-2V(z)}=z^ae^{-z}$, supported on $[0,\infty)$.
The Pearson pair of the Laguerre weight is $(f,g)=(z,z-a)$, and therefore, according to 
(\ref{214})
the external potentials are $e^{-V_1(z)}=z^{(a-1)/2}e^{-z/2}$ and $e^{-V_2(z)}=z^{(a+1)/2}e^{-z/2}$.

For the monic Laguerre polynomials \footnote{The monic Laguerre polynomials are given in terms of the Laguerre polynomials by $p_k^{(a)}(z)=(-1)^kk!L_k^{(a)}(z)$} $\{p_k^{(a)}(z)\}_{k=0}^\infty$, we have
\begin{align}\label{laguerre}
\begin{aligned}
&\langle p_m^{(a)},p_n^{(a)}\rangle_{2,V}=\Gamma(n+1)\Gamma(n+a+1)\delta_{n,m},\\
&p_n^{(a-1)}(z)=p_n^{(a)}(z)+np_{n-1}^{(\al)}(z),\quad \frac{d}{dx}p_n^{(a)}(z)=np_{n-1}^{(a+1)}(z).
\end{aligned}
\end{align}
From \cite[Equation 3.6]{adler00} or \cite[Equation 4.31]{nagao07}, one knows
\begin{align*}
\langle p_{2m}^{(a)},p_{2n+1}^{(a)}-2n(2n+a)p_{2n-1}^{(a)}\rangle_{1,V_1}=2\Gamma(2n+1)\Gamma(2n+a+1)\delta_{n,m}.
\end{align*}
Properties of the Laguerre polynomials \eqref{laguerre} tell us
\begin{align*}
&\langle p_{2m}^{(a)},p_{2n+1}^{(a)}\rangle_{4,V_2}=\frac{1}{2} (2n+1)\langle p_{2m}^{(a)},p_{2n+1}^{(a+1)}\rangle_{2,V_2}-\frac{1}{2} (2m)\langle p_{2m-1}^{(a+1)},p_{2n}^{(a)}\rangle_{2,V_2}\\
&=\frac{1}{2} (2n+1)\langle p_{2m}^{(a+1)}+2mp_{2m-1}^{(a+1)},p_{2n}^{(a+1)}\rangle_{2,V_2}-\frac{1}{2} (2m) \langle p_{2m-1}^{(a+1)},p_{2n+1}^{(a+1)}+(2n+1)p_{2n}^{(a+1)}\rangle_{2,V_2}\\
&=\frac{1}{2}\Gamma(2n+2)\Gamma(2n+a+2)\delta_{n,m}-\frac{1}{2}\Gamma(2m+1)\Gamma(2m+a+1)\delta_{m-1,n}.
\end{align*}
Therefore, the coefficients in \eqref{rel1} for the Laguerre case  are
\begin{align}\label{lcoe}
\zeta_j=2j(2j+a), \quad \ho_j=2\Gamma(2j+1)\Gamma(2j+a+1),\quad \hf_j=\frac{1}{2}\Gamma(j+2)\Gamma(j+a+2).
\end{align}
Substituting in the recurrence relation \eqref{ra} shows
\begin{align*}
(2j+1)&(2j+2)(2j+a+1)(2j+a+2)\al_{j+1}\\&-(4X^2+(2j+1)(2j+a+1)+2j(2j+a))\al_j+\al_{j-1}=0.
\end{align*}
Scaling $\al_{j}=\frac{1}{(2j)!}\hat{\al}_j$ gives that $\{\hat{\al}_j\}_{j=0}^\infty$  satisfy
\begin{align*}
(j+\frac{a}{2}+1)(j+\frac{a}{2}+\frac{1}{2})\hat{\al}_{j+1}-\left({X^2}+(j+\frac{1}{2})(j+\frac{a}{2}+\frac{1}{2})+j(j+\frac{a}{2})\right)\hat{\al}_j+j(j-\frac{1}{2})\hat{\al}_{j-1}=0.
\end{align*}
Together with the initial condition (\ref{da1}), we recognise this three term recurrence as that satisfied by a particular
continuous dual Hahn polynomial \cite[Section 1.3]{koekoek96} 
\begin{align*}
\hat{\al}_j=\tilde{S}_j\left(-{X^2};\frac{a+1}{2},\frac{1}{2},0\right)
\end{align*}
or hypergeometric function
\begin{align*}
\hat{\al}_j={_3}F_2\!\left(\left.\begin{array}{c}
-j,\, \frac{a+1}{2}-X,\, \frac{a+1}{2}+X\\
\frac{a+2}{2},\,\frac{a+1}{2}
\end{array}\right|1
\right).
\end{align*}
Hence, for the coefficients $\{\al_j\}_{j=0}^\infty$ we have
\begin{align*}
\al_j=\frac{1}{(2j)!}\tilde{S}_j\left(-{X^2};\frac{a+1}{2},\frac{1}{2},0\right).
\end{align*}
As for  the Laguerre case we can check 
$
\zeta_{j+1}\hf_{2j}=\hf_{2j+1},
$
and thus the recurrence relations of $\{\xi_j\}_{j=0}^\infty$ are the same with $\{\al_j\}_{j=0}^\infty$. Therefore,
\begin{align*}
\xi_j=\frac{1}{(2j)!}\tilde{S}_j\left(-X^2;\frac{a+1}{2},\frac{1}{2},0\right).
\end{align*}

\subsection{Jacobi case}\label{jacobi}
The weight function in the Jacobi case is  $e^{-2V(z)}=(1-z)^a(1+z)^b$ supported on the interval $(-1,1)$. This is more general than the Gaussian and Laguerre weights, in the
sense that the latter can be obtained as limiting cases. 
For the Pearson pair, we have $(f,g)=\left(1-z^2,(a+b)z+(a-b)\right)$, and therefore the external potentials are $e^{-V_1(z)}=(1-z)^{(a-1)/2}(1+z)^{(b-1)/2}$ and $e^{-V_2(z)}=(1-z)^{(a+1)/2}(1+z)^{(b+1)/2}$.

The monic Jacobi polynomials $\{p_k^{(a,b)}(z)\}_{k=0}^\infty$ have the well known  properties
\begin{align}\label{jacobi}
\begin{aligned}
&\langle p_m^{(a,b)},p_n^{(a,b)}\rangle_{2,V}=h_n^{(a,b)}\delta_{n,m},\quad\frac{d}{dz}p_n^{(a,b)}(z)=np_{n-1}^{(a+1,b+1)}(z),\\
&p_n^{(a,b)}(z)=p_n^{(a+1,b+1)}(z)+\mu_n^{(a,b)}p_{n-1}^{(a+1,b+1)}(z)+\epsilon_n^{(a,b)}p_{n-2}^{(a+1,b+1)}(z),
\end{aligned}
\end{align}
where the coefficients are 
\begin{align}\label{coeff}
\begin{aligned}
h_n^{(a,b)}&=2^{a+b+2n+1}\frac{\Gamma(n+1)\Gamma(n+a+1)\Gamma(n+b+1)\Gamma(n+a+b+1)}{\Gamma(2n+a+b+1)\Gamma(2n+a+b+2)},\\
\mu_n^{(a,b)}&=\frac{2n}{(2n+a+b+1)}\left(-\frac{n+b+1}{2n+a+b+2}+\frac{n+a}{2n+a+b}\right),\\
\epsilon_n^{(a,b)}&=-\frac{4n(n-1)(n+a)(n+b)}{(2n+a+b+1)(2n+a+b)^2(2n+a+b-1)}.
\end{aligned}
\end{align}
From \cite[Equation 3.14]{adler00}, one knows
\begin{align*}
\langle p_{2m}^{(a,b)},p_{2n+1}^{(a,b)}-\zeta_n^{(a,b)}p_{2n-1}^{(a,b)}\rangle_{1,V_1}=\ho_n\delta_{n,m}
\end{align*}
with
\begin{align*}
\zeta_n^{(a,b)}&=\frac{8n(2n+a)(2n+b)(2n+a+b)}{(4n+a+b-1)(4n+a+b)(4n+a+b+1)(4n+a+b+2)},\\
\ho_n&=2^{a+b+4n+2}\frac{\Gamma(2n+1)\Gamma(2n+a+1)\Gamma(2n+b+1)\Gamma(2n+a+b+1)}{\Gamma(4n+a+b+1)\Gamma(4n+a+b+3)}.
\end{align*}
Furthermore, from the properties \eqref{jacobi}, one can compute
\begin{align}\label{hfj}
\begin{aligned}
&\langle p_{2m}^{(a,b)},p_{2n+1}^{(a,b)}\rangle_{4,V_2}=\frac{1}{2}(2n+1)\langle p_{2m}^{(a,b)},p_{2n}^{(a+1,b+1)}\rangle_{2,V_2}-\frac{1}{2}(2m)\langle p_{2m-1}^{(a+1,b+1)},p_{2n+1}^{(a,b)}\rangle_{2,V_2}\\
&=\frac{1}{2}(2n+1)\langle p_{2m}^{(a+1,b+1)}+\mu^{(a,b)}_{2m}p_{2m-1}^{(a+1,b+1)}+\epsilon^{(a,b)}_{2m}p_{2m-2}^{(a+1,b+1)},p_{2n}^{(a+1,b+1)}\rangle_{2,V_2}\\
&\quad-\frac{1}{2}(2m)\langle p_{2m-1}^{(a+1,b+1)},p_{2n+1}^{(a+1,b+1)}+\mu^{(a,b)}_{2n+1}p_{2n}^{(a+1,b+1)}+\epsilon^{(a,b)}_{2n+1}p_{2n-1}^{(a+1,b+1)}\rangle_{2,V_2}\\
&=\frac{1}{2}\left(
(2n+1)h_{2n}^{(a+1,b+1)}-(2n)\epsilon^{(a,b)}_{2n+1}h_{2n-1}^{(a+1,b+1)}
\right)\delta_{n,m}\\
&\quad -\frac{1}{2}\left(
(2m)h_{2m-1}^{(a+1,b+1)}-(2m-1)\epsilon^{(a,b)}_{2m}h_{2m-2}^{(a+1,b+1)}
\right)\delta_{m-1,n},
\end{aligned}
\end{align}
where $h_n^{(a,b)}$ and $\epsilon_n^{(a,b)}$ are defined in \eqref{coeff}. As a consequence,
\begin{align*}
\zeta_j&=\frac{8j(2j+a)(2j+b)(2j+a+b)}{(4j+a+b-1)(4j+a+b)(4j+a+b+1)(4j+a+b+2)},\\
\ho_j&=2^{a+b+4j+2}\frac{\Gamma(2j+1)\Gamma(2j+a+1)\Gamma(2j+b+1)\Gamma(2j+a+b+1)}{\Gamma(4j+a+b+1)\Gamma(4j+a+b+3)},\\
\hf_j&=2^{a+b+2j+2}\frac{\Gamma(j+2)\Gamma(j+a+2)\Gamma(j+b+2)\Gamma(j+a+b+2)}{\Gamma(2j+a+b+2)\Gamma(2j+a+b+4)}.
\end{align*}
Substituting  into the recurrence relation \eqref{ra}, one finds that the coefficients $\{\al_j\}_{j=0}^\infty$ satisfy the recurrence
\begin{align*}
&\left(X^2+\frac{(2j+1)(2j+a+1)(2j+b+1)(2j+a+b+1)}{(4j+a+b+1)(4j+a+b+3)}+\frac{2j(2j+a)(2j+b)(2j+a+b)}{(4j+a+b-1)(4j+a+b+1)}\right)\al_j\\
&=4\frac{(2j+1)_2(2j+a+1)_2(2j+b+1)_2(2j+a+b+1)_2}{(4j+a+b+1)_3(4j+a+b+3)_3}\al_{j+1}+\frac{1}{4}{(4j+a+b-2)(4j+a+b)}\al_{j-1},
\end{align*}
where the Pochhammer symbol $(a)_k$ is defined by
\begin{align*}
(a)_0=1,\quad \text{and}\quad (a)_k=a(a+1)\cdots(a+k-1),\, k=1,2,3,\cdots.
\end{align*} 
Making the substitution
\begin{align*}
\al_j=2^{-2j}\frac{(4j+a+b+1)!}{(2j)!(2j+b)!}\hat{\al}_{j}
\end{align*}
gives the modified recurrence
\begin{align*}
&\left(\frac{X^2}{4}+\frac{(j+\frac{1}{2})(j+\frac{a+1}{2})(j+\frac{b+1}{2})(j+\frac{a+b+1}{2})}{(2j+\frac{a+b+3}{2})(2j+\frac{a+b+1}{2})}+\frac{j(j+\frac{a}{2})(j+\frac{b}{2})(j+\frac{a+b}{2})}{(2j+\frac{a+b+1}{2})(2j+\frac{a+b-1}{2})}\right)\hat{\al}_j\\
&
=\frac{(j+\frac{a+2}{2})(j+\frac{a+1}{2})(j+\frac{a+b+2}{2})(j+\frac{a+b+1}{2})}{(2j+\frac{a+b+3}{2})(2j+\frac{a+b+1}{2})}\hat{\al}_{j+1}+\frac{j(j-\frac{1}{2})(j+\frac{b}{2})(j+\frac{b-1}{2})}{(2j+\frac{a+b+1}{2})(2j+\frac{a+b-1}{2})}\hat{\al}_{j-1}.
\end{align*}

This modified recurrence, together with the initial conditions (\ref{da1}), is recognised as being satisfied by a
special case of the Wilson polynomials \cite[Section 1.1]{koekoek96}, telling us that
\begin{align*}
\hat{\al}_j=\tilde{W}_j\left(-\frac{X^2}{4};\frac{a+1}{2},\frac{b+1}{2},\frac{1}{2},0\right)
\end{align*}
or in terms of a hypergeometric function
\begin{align*}
\hat{\al}_j={_4}F_3\!\left(
\left.\begin{array}{c}
-j,\, j+\frac{a+b+1}{2},\, \frac{a+1-X}{2},\, \frac{a+1+X}{2}\\
\frac{a+1}{2},\,\frac{a+2}{2},\,\frac{a+b+2}{2}
\end{array}\right|1
\right).
\end{align*}
Hence
\begin{align*}
\al_j=2^{-2j}\frac{(4j+a+b+1)!}{(2j)!(2j+b)!}\tilde{W}_j\left(
-\frac{X^2}{4};\frac{a+1}{2},\frac{b+1}{2},\frac{1}{2},0
\right).
\end{align*}

In contrast to the Gaussian and Laguerre cases, in the Jacobi case $\zeta_j\hf_{2j}\not=\hf_{2j+1}$.
This means the recurrence relation for $\{\xi_j\}_{j=0}^\infty$ is different to that for $\{\al_j\}_{j=0}^\infty$.
However, one can make use of Corollary \ref{cjj} to conclude that if we change $\hat{\xi}_j=c_j\xi_j$, then $\{\hat{\xi}_j\}_{j=0}^\infty$ admits the same recurrence as $\{\al_j\}_{j=0}^\infty$. 
In the Jacobi  case, $c_j$ can be computed via \eqref{cj} and we find
\begin{align*}
c_j=\prod_{k=1}^j \zeta_k\frac{\hf_{2k-2}}{\hf_{2k-1}}=\prod_{k=1}^j \frac{4k+a+b-2}{4k+a+b+2}=\frac{2+a+b}{4j+a+b+2}.
\end{align*}
Therefore, $\xi_j=c_j^{-1}\al_j$ and 
\begin{align*}
\xi_j=2^{-2j}\frac{(4j+a+b+2)!}{(2j)!(2j+b)!(2+a+b)}\tilde{W}_j\left(-\frac{X^2}{4};\frac{a+1}{2},\frac{b+1}{2},\frac{1}{2},0\right).
\end{align*}

\subsection{Generalised Cauchy case}\label{3.4}
In this subsection, we consider the remaining of the classical weights on the line, namely the
generalised Cauchy weight $e^{-2V(z)}=(1-iz)^c(1+iz)^{\bar{c}}$ with $z\in\mathbb{R}$, $c\in\mathbb{C}$. In the one-component case, this weight is in fact related to the circular Jacobi ensemble \cite{liu17,forrester01}, whose probability density function (PDF) is defined on $[0,2\pi)^N$. Thus if we consider the stereographic projection
\begin{align*}
e^{i\theta_j}=\frac{1-ix_j}{1+ix_j},
\end{align*}
then we can transform the eigenvalue PDF on the unit circle
\begin{align*}
\prod_{j=1}^N (1-e^{i\theta_j})^{\beta a}(1-e^{-i\theta_j})^{\beta \bar{a}}\prod_{1\leq j<k\leq N}|e^{i\theta_k}-e^{i\theta_j}|^{\beta}d\theta_1\cdots d\theta_N
\end{align*}
 to the  eigenvalue PDF on the real line
\begin{align*}
\prod_{j=1}^N (1-ix_j)^{c}(1+ix_j)^{\bar{c}}\prod_{1\leq j<k\leq N}|x_k-x_j|^\beta dx_1\cdots dx_N.
\end{align*}

Denoting the real part and imaginary part of $c$ as $\Re(c)=-p$ and $\Im(c)=q$ respectively, the Pearson pair for the generalised Cauchy weight is seen to be $(f,g)=\left(1+z^2,2pz+2iq\right)$. Therefore, the external potentials of the
two-component model are given by $e^{-V_1(z)}=(1-iz)^{(c-1)/2}(1+iz)^{(\bar{c}-1)/2}$ and $e^{-V_2(z)}=(1-iz)^{(c+1)/2}(1+iz)^{(\bar{c}+1)/2}$. We remark that the orthogonal polynomials with 
the generalised Cauchy weight are usually referred to as (a special case of) the  Routh-Romanovski polynomials \cite{masjed12}. Since these polynomials are not well known, we give a brief introduction in the appendix, together with the construction of the corresponding skew orthogonal polynomials in the case $\beta=1$.

Specifically, we give in \eqref{rr1} the explicit forms of the $\zeta_j$ and $\ho_j$ for the  Routh-Romanovski polynomials.
Moreover, from the relationship between Routh-Romanovski polynomials and Jacobi polynomials \eqref{rj}, and the properties of Jacobi polynomials \eqref{jacobi}, we know 
\begin{align*}
\I_n^{(c,\bar{c})}(z)=\I_n^{(c+1,\bar{c}+1)}(z)+\mu_n^{c}\I^{(c+1,\bar{c}+1)}_{n-1}(z)+\epsilon_n^c \I_{n-2}^{(c+1,\bar{c}+1)}(z)
\end{align*}
with coefficients $\mu_n^c=i\mu_n^{(a,b)}|_{a\to c,\,b\to \bar{c}}$ and $\epsilon_n^c=i^2\epsilon_n^{(a,b)}|_{a\to c,\, b\to\bar{c}}$, where $\mu_n^{(a,b)}$ and $\epsilon_n^{(a,b)}$ are given in \eqref{coeff}.
Therefore, in analogy with the  computation in \eqref{hfj}, $\hf_{2j}$ in the Routh-Romanovski case can be computed as
\begin{align*}
\hf_{2j}&=\frac{1}{2}\left((2j+1)h_{2j}^{(c-1,\bar{c}-1)}-2n\epsilon_{2j+1}^c h_{2j-1}^{(c-1,\bar{c}-1)}\right)\\
&=2^{4j-2p+3}\pi \frac{\Gamma(2j+2)\Gamma(2p-4j-1)\Gamma(2p-4j-3)}{\Gamma(2p-2j-1)\Gamma(p-1-iq-2j)\Gamma(p-1+iq-2j)}.
\end{align*}
Taking $\zeta_j$, $\ho_j$ and $\hf_j$ into the equation \eqref{ra} and dividing $\hf_j$ on both sides, one can obtain
\begin{align*}
\left(X^2+\frac{(2j+1)(2p-2j-1)(-c-2j-1)(-\bar{c}-2j-1)}{(2p-4j-1)(2p-4j-3)}+\frac{2j(2p-2j)(-c-2j)(-\bar{c}-2j)}{(2p-4j-1)(2p-4j+1)}\right)\al_j\\
=2^2\frac{(2j+1)_2(2p-2j-2)_2(-c-2j-2)_2(-\bar{c}-2j-2)_2}{(2p-4j-3)_3(2p-4j-5)_3}\al_{j+1}+2^{-2}(2p-4j)(2p-4j+2)\al_{j-1}.
\end{align*}
If we make the change of variable 
\begin{align*}
\al_j=(-1)^j2^{-2j}\frac{(-c-2j-1)!}{(2j)!(2p-4j-2)!}\hat{\al}_j,
\end{align*}
the above equation is transformed into
\begin{align*}
&\left(-\frac{X^2}{4}+\frac{(j+\frac{1}{2})(j-p+\frac{1}{2})(j+\frac{c+1}{2})(j+\frac{\bar{c}+1}{2})}{(2j-p+\frac{3}{2})(2j-p+\frac{1}{2})}+\frac{j(j-p)(j+\frac{c}{2})(j+\frac{\bar{c}}{2})}{(2j-p+\frac{1}{2})(2j-p-\frac{1}{2})}
\right)\hat{\al}_j\\&=\frac{(j-p+1)(j-p+\frac{1}{2})(j+\frac{\bar{c}+1}{2})(j+\frac{\bar{c}+2}{2})}{(2j-p+\frac{3}{2})(2j-p+\frac{1}{2})}\hat{\al}_{j+1}+\frac{j(j-\frac{1}{2})(j+\frac{c}{2})(j+\frac{c-1}{2})}{(2j-p+\frac{1}{2})(2j-p-\frac{1}{2})}\hat{\al}_{j-1}.
\end{align*}
The solution of this recurrence relation can also be expressed as a Wilson polynomial
\begin{align*}
\hat{\al}_j=\tilde{W}_j\left(
\frac{X^2}{4};\frac{\bar{c}+1}{2},\frac{c+1}{2},\frac{1}{2},0
\right)
\end{align*}
in analogy with the Jacobi case, and hence 
\begin{align*}
\al_j=(-1)^j2^{-2j}\frac{\Gamma(-c-2j)}{\Gamma(2j+1)\Gamma(2p-4j-1)}\tilde{W}_j\left(
\frac{X^2}{4};\frac{\bar{c}+1}{2},\frac{c+1}{2},\frac{1}{2},0
\right).
\end{align*}
The computation of $\{\xi_j\}$ is based on the Corollary \ref{cjj} and equation \eqref{cj}. One can show
\begin{align*}
\xi_j=\frac{2p-4j-2}{2p-2}\al_j=\frac{(-1)^j 2^{-2j}\Gamma(-c-2j)}{(2p-2)\Gamma(2j+1)\Gamma(2p-4j-2)}\tilde{W}_j\left(
\frac{X^2}{4};\frac{\bar{c}+1}{2},\frac{c+1}{2},\frac{1}{2},0
\right).
\end{align*}

\section{Concluding remarks}
The focus of the paper has been on the skew orthogonal polynomial theory arising from the two-component plasma for the mixed charge on the line in the case of classical weight functions,
and so on giving a unified framework including the Gaussian case considered in \cite{Rider12}. It has been shown that for all of the classical cases the skew orthogonal polynomials can be written as a linear combination of the classical orthogonal polynomials, with coefficients and partition functions being expressed in terms of Askey-scheme hypergeometric orthogonal polynomials, which we summarise in Table \ref{t1}.
\begin{table}[htbp]
\caption{Coefficients with regard to classical weights}
\begin{tabular}{cc}\hline\hline
Classical weights &  Coefficients/Partition functions\\ \hline
Gaussian & Laguerre polynomials ($_1F_1$)\\
Laguerre & continuous Hahn polynomials ($_3 F_2$)\\
Jacobi & Wilson polynomials ($_4F_3$)\\
Generalised Cauchy & Wilson polynomials ($_4F_3$)\\
\hline\hline
\end{tabular}
\label{t1}
\end{table}

Still, a number of questions remain for further investigation:
\begin{enumerate}
\item In random matrix theory, one motivation for knowledge of the expansion of the classical skew orthogonal polynomials in terms of classical orthogonal polynomials 
is that the Christoffel-Darboux kernel for the skew orthogonal polynomials has a simpler form in terms of orthogonal polynomials,
as shown in \cite{widom99,adler00}; see also \cite{Forrester06} in the setting of a combinatorial model. However, an analogous evaluation of the kernel for the two-component plasma on a line is not known;
\item In the present paper, we only consider the case that the measure set is continuously supported in $\mathbb{R}$. One would like to develop a theory of the two-component plasma on the discrete (or $q$-) lattice, to include (perhaps) a broader class of the Askey-Wilson scheme. The difficulty here is to find the discrete (or $q$-) skew orthogonal polynomials with $\beta=1$, and thus to obtain the first two equations in relation \eqref{rel1} in a discrete setting (Note: the second two equations are naturally satisfied due to the discrete Pearson equation). Although a discrete orthogonal ensemble is given by \cite{borodin09}, or in terms of discrete Selberg integral by \cite{borodin16}, the skew orthogonal polynomials with $\beta=1$ are still unknown to us;
\item The model we consider in this paper can be viewed as a mixture of $\beta=1$ and $\beta=4$ ensembles in random matrix theory. Another known mixture model, defined on the circle and now interpolating between $\beta=2$ and $\beta=4$, is given by the Boltzmann factor
\begin{align*}
\prod_{1\leq j<k\leq N_1}|e^{i\theta_k}-e^{i\theta_j}|^2\prod_{1\leq\al<\be\leq N_2}|e^{i\phi_\al}-e^{i\phi_\be}|^4\prod_{j=1}^{N_1}\prod_{\al=1}^{N_2}|e^{i\theta_j}-e^{i\phi_\al}|^2.
\end{align*}
This is relevant to studies of the quantum Hall effect \cite{forrester842,forrester11}. Although this model has been shown to be solvable, and the skew orthogonal polynomials exhibited in \cite{forrester11}, it remains to similarly analyse the analogue of this mixture model on the real line.
\end{enumerate}

\section*{Acknowledgements} 
This work is part of a research program supported by the Australian Research Council (ARC) through the ARC Centre of Excellence for Mathematical and Statistical frontiers (ACEMS). PJF also acknowledges partial support from ARC grant DP170102028. S.-H. Li would like to thank ZiF, University of Bielefeld for their generous support when he attended the summer school `Randomness in Physics and Mathematics: From Stochastic Processes to Networks', where this work was begun.

\begin{appendix}
\section{An introduction to the Routh-Romanovski polynomials}

The Routh-Romanovski polynomials were introduced by Routh \cite{routh84} and Romanovski \cite{romanovski29} and have subsequently shown to be of relevance to several studies in physics \cite{Ra07}. This family of polynomials are orthogonal with respect to the weight function 
\begin{align*}
W^{(p,q)}(z;a,b,c,d)=\left(
(az+b)^2+(cz+d)^2
\right)^{-p}\exp\left(
2q\, \arctan \frac{az+b}{cz+d}
\right).
\end{align*}
In this paper, we restrict ourselves to the case 
\begin{align*}
e^{-2V(z)}:=(1-iz)^{c}(1+iz)^{\bar{c}}=(1+z^2)^{\Re(c)}\exp(2\Im(c)\arctan (z))=W^{\left(\Re(c), \Im(c)\right)}(z;1,0,0,1),
\end{align*}
where $\Re(c)$ and $\Im(c)$ denote the real and imaginary part of $c$.
The weight function is then simply related to the Jacobi weight and one can show that correspondingly
the monic Routh-Romanovski polynomials $\{\I_n^{(c,\bar{c})}(z)\}_{n=0}^\infty$ can be connected with the monic Jacobi polynomials (see Section \ref{jacobi}) via the relation
\begin{align}\label{rj}
\I_n^{(c,\bar{c})}(z)=i^{-n} p_n^{(c,\bar{c})}(iz).
\end{align}
If we denote $c=-p+iq$, then we have the orthogonal relation
\begin{align*}
\langle \I_n^{(c,\bar{c})}, \I_m^{(c,\bar{c})}\rangle_{2,V}=2^{2n-2p+2}\pi\frac{\Gamma(n+1)\Gamma(2p-2n)\Gamma(2p-2n-1)}{\Gamma(2p-n)\Gamma(p-iq-n)\Gamma(p+iq-n)}\delta_{n,m}:=h_n^{(c,\bar{c})}\delta_{n,m}.
\end{align*}
It should be noted that this family of orthogonal polynomials are valid for $n\in\mathbb{N}$ and $n<p$ only. 
The Pearson pair of this weight can be computed as 
\begin{align}\label{prr}
(f,g)=(1+z^2, 2pz+2iq).
\end{align}

\section{Construction of skew orthogonal polynomials with $\beta=1$ relating to the  Routh-Romanovski polynomials}\label{sop1}
A procedure to construct the classical skew orthogonal polynomials for the Hermite, Laguerre and Jacobi weights with $\beta=1,\, 4$ has been given in \cite{adler00}; a comprehensive review can be found in \cite[\S 6]{forrester10}. Later, the remaining classical weight function --- the generalised Cauchy weight --- was considered in \cite{forrester01} during studies of the circular ensemble. In this appendix, we will give a brief review of the
construction of the skew orthogonal polynomials with $\beta=1$ in the classical cases and apply it to the generalised Cauchy weight, to obtain the relations \eqref{rel1} in the Routh-Romanovski case.
Throughout the equations (\ref{214}) will be assumed.

For the classical weights, we can start with the Pearson pair $(f,g)$ and construct an operator $\mathcal{A}:=f\partial_z+(f'-g)/2$ of order at most $1$, such that
\begin{align*}
\mathcal{A}p_j(z)=-\frac{c_j}{\langle p_{j+1},p_{j+1}\rangle_{2,V}}p_{j+1}(z)+\frac{c_{j-1}}{\langle p_{j-1},p_{j-1}\rangle_{2,V}}p_{j-1}(z).
\end{align*}
This operator has the significant property of inter-relating the $\beta = 2$ inner product with the $\beta = 1$ and $\beta = 4$ skew inner products,
\begin{align*}
\langle \phi,\mathcal{A}\psi\rangle_{2,V}=\langle \phi,\psi\rangle_{4,V_2},\quad \langle \phi,\mathcal{A}^{-1}\psi\rangle_{2,V}=-\langle \phi,\psi\rangle_{1,V_1}.
\end{align*}
Specific attention is paid to the case $\beta = 1$. Denoting $\gamma_j=c_j/(\langle p_{j+1},p_{j+1}\rangle_{2,V}\langle p_j,p_j\rangle_{2,V})$, one finds the relation between the monic orthogonal polynomials $\{p_{j}(z)\}_{j=0}^\infty$ and the skew orthogonal polynomials $\{q_j(z)\}_{j=0}^\infty$ with $\beta=1$ as
\begin{align*}
q_{2j}(z)=p_{2j}(z),\quad q_{2j+1}(z)=p_{2j+1}(z)-\frac{\gamma_{2j-1}}{\gamma_{2j}}p_{2j-1}(z),
\end{align*}
and for the normalisation $\langle q_{2j}(z),q_{2j+1}(z)\rangle_{1,V_1}=1/\gamma_{2j}$. 

Therefore, from the Pearson pair \eqref{prr} of Routh-Romanovski polynomials, we can construct $\mathcal{A}=(1+z^2)\partial_z+(1-p)z-iq$. Since we consider the monic Routh-Romanovski polynomials throughout the paper, it is easy to compute the coefficients 
$
c_j=(p-1-j)\langle \I^{(c,\bar{c})}_{j+1},\I^{(c,\bar{c})}_{j+1}\rangle_{2,V}.
$
Therefore, one can obtain
\begin{align*}
\langle \I^{(c,\bar{c})}_{2m},\I^{(c,\bar{c})}_{2n+1}-\zeta_n\I^{(c,\bar{c})}_{2n-1}\rangle_{1,V_1}=\ho_n\delta_{n,m}
\end{align*}
where
\begin{align}\label{rr1}
\begin{aligned}
\zeta_n&=\frac{16n(p-n)(p-iq-2n)(p+iq-2n)}{(2p-4n-2)(2p-4n-1)(2p-4n)(2p-4n+1)},\\
\ho_n&=2^{4n-2p+3}\pi\frac{\Gamma(2n+1)\Gamma(2p-4n)\Gamma(2p-4n-2)}{\Gamma(2p-2n)\Gamma(p+iq-2n)\Gamma(p-iq-2n)}.
\end{aligned}
\end{align}

\end{appendix}

\small
\bibliographystyle{abbrv}

\end{document}